\documentclass{article}

\usepackage{amssymb}
\usepackage{lipsum}
\usepackage{amsthm}

\usepackage{amsfonts}
\usepackage{mathrsfs}
\usepackage{systeme, mathtools}
\usepackage{stmaryrd}

\usepackage{tikz-cd}
\usepackage{xcolor}
\usepackage{graphicx}
\usepackage{float}
\usepackage{dcolumn}
\usepackage{bm}
\usepackage{appendix}
\usepackage{url}
\usepackage{hyperref}
\hypersetup{hidelinks}
\usepackage{authblk}

\numberwithin{equation}{section}
\newtheorem{theorem}{Theorem}[section]
\newtheorem{proposition}{Proposition}[section]
\newtheorem{lemma}[proposition]{Lemma}
\newtheorem{corollary}[proposition]{Corollary}

\theoremstyle{definition}
\newtheorem*{remark}{Remark}

\newcommand{\R}{\mathbb{R}}

\newcommand{\g}{\mathfrak{g}}
\newcommand{\A}{\mathcal{A}}
\newcommand{\su}[1]{\mathfrak{su}(#1)}
\newcommand{\so}[1]{\mathfrak{so}(#1)}
\newcommand{\M}{\mathcal{M}}
\newcommand{\T}{\mathrm{T}}

\newcommand{\OF}[1]{P^{SO}(#1)}
\newcommand{\SP}[1]{P^{Spin}(#1)}

\newcommand{\Ad}{\operatorname{Ad}}

\begin{document}

\title{Moduli space of spin connections on three-dimensional homogeneous spaces}

\author[1,3]{Matteo Bruno\footnote{\texttt{matteo.bruno@uniroma1.it }}}
\author[2]{Gabriele Peluso\footnote{\texttt{gabriele.peluso@uniroma1.it }}}

\affil[1]{\textit{\normalsize{Physics Department, Sapienza University of Rome, P.za Aldo Moro 5, Rome, 00185, , Italy}}}
\affil[2]{\textit{\normalsize{Math Department, Sapienza University of Rome, P.za Aldo Moro 5, Rome, 00185, , Italy}}}
\affil[3]{\textit{\normalsize{INFN, Sezione di Roma 1, P.le Aldo Moro 2, 00185, Rome, Italy}}}
\date{}

\maketitle

\begin{abstract}
In this manuscript, we aim to classify and characterize the moduli space of homogeneous spin connections and homogeneous $SU(2)$ connections on three-dimensional Riemannian homogeneous spaces. An analysis of the topology of the associated moduli spaces reveals that they are finite-dimensional topological manifolds (possibly with boundary) possessing trivial homotopy groups.

\medskip

Owing to their deep connection with cosmological models in the\\ Ashtekar-Barbero-Immirzi formulation of General Relativity, this study offers a mathematically rigorous interpretation of the Ashtekar-Barbero-Immirzi-Sen connection within a cosmological context. In particular, we show that a correct formulation of the theory relies crucially on identifying the moduli space of homogeneous spin connections, thereby emphasizing the essential role of the spin structure in ensuring consistency with the physical content of the theory.

\medskip

The favorable topological properties of these moduli spaces circumvent many of the usual difficulties associated with singularities and the definition of regular measures in Quantum Field Theory and Quantum Gravity. As a result, they provide a solid foundation for the rigorous implementation of quantum theory in the cosmological setting.
\end{abstract}

\section{Introduction}
\label{sec:Intro}

The study of 3-dimensional Riemannian homogeneous spaces has been a significant topic in mathematics due to its connection with Thurston's geometrization conjecture \cite{Thurston_1982}. These spaces also play a major role in General Relativity, as they are used to describe cosmological models, thereby contributing to the development of Cosmology within mathematical General Relativity, in accordance with the cosmological principle \cite{Landau_Lifshits_1975}. The link between these two fields is strong: the geometries in Thurston's conjecture correspond to the most relevant cosmological models, commonly known as class A models \cite{Milnor_1976,Patrangenaru_1996}.

\medskip

In this work, we aim to clarify the relationship between 3-dimensional Riemannian homogeneous spaces and a particular formulation of General Relativity, known as the Ashtekar-Barbero-Immirzi formulation \cite{Ashtekar_1986,Ashtekar_1987,Barbero_1995,Immirzi_1997}. In this formulation, the usual initial data set $(M, g, k)$, where $M$ is a Riemannian 3-manifold, $g$ is the Riemannian metric on $M$, and $k$ is a symmetric $(0,2)$-tensor, is replaced by $(M, g, \omega)$, where $M$ is a spin manifold and $\omega$ is a spin connection. The formulation also establishes a correspondence between the two datasets \cite{Bruno_2024a}. In this way, General Relativity takes the form of an $SU(2)$ gauge theory.

\medskip

While the implications of homogeneity for $(M, g, k)$ are well established, the extension of this property to $(M, g, \omega)$ remains a topic of debate, particularly regarding the meaning of homogeneity for $\omega$. A recent work \cite{Bruno_2024b} demonstrated that interpreting $\omega$ as a homogeneous spin connection (via Wang's theorem \cite{Wang_1958}) yields the standard notion of a homogeneous $k$. An earlier work \cite{Bojowald_Kastrup_2000} proposed interpreting $\omega$ as an homogeneous $SU(2)$ connection. While in the non-homogeneous case, spin connections and $SU(2)$ connections are equivalent on a spin 3-manifold, this equivalence does not hold in the homogeneous case. We will demonstrate this by classifying both types of connections and comparing the results with those of the physical theory. Furthermore, we compute the respective moduli spaces for each 3-dimensional homogeneous space.

\medskip

The moduli space is a critical object in the quantization of gauge field theories. Typically, moduli spaces exhibit certain singularities; however, in our case, the moduli spaces of both types of connections exhibit remarkably well-behaved properties:
\begin{itemize}
    \item The moduli space depends only on the isotropy group of the homogeneous space, corresponding to the symmetry class of the associated cosmological model.
    \item The moduli space admits the structure of a finite-dimensional topological manifold (with boundary).
\end{itemize}

While the first property aligns with expectations from physical computations, the second introduces a significant regularization in the quantum theory of the homogeneous Ashtekar-Barbero-Immirzi formulation. This regularization arises from the ability to employ the standard, well-defined tools used in the quantization of systems with a finite number of degrees of freedom, thereby opening the way for a rigorous mathematical investigation of the corresponding quantum theory.

\medskip

Finally, the explicit computation of the moduli spaces reveals a clear and meaningful distinction between the two types of connections. This distinction allows us to conclude that the Ashtekar–Barbero–Immirzi–Sen connection, when interpreted in the cosmological setting, should be understood as a homogeneous spin connection (up to a slight generalization, as discussed below). This interpretation not only aligns with the mathematical structure of the moduli space but also agrees with the physical requirements of the theory, reinforcing the necessity of incorporating spin structures in any rigorous approach to the Ashtekar-Barbero-Immirzi formulation.

\section{Homogeneous spaces, spin structures and connections}
\label{sec:Rev}
In this section, we want to introduce the structure and objects that will be the main characters of our study and some well-known facts about them. A \textit{Riemannian homogeneous space} is a pair $(M,G)$ where $M$ is an $n$-dimensional orientable Riemannian manifold and $G$ a connected Lie group which acts transitively from the left via orientation-preserving isometries. Fixing a point $o\in M$, the stabilizer of that point is the so-called \textit{isotropy subgroup} $H\subset G$. Whenever $H$ is connected and compact, $M\simeq G/H$ is reductive, namely $\mathfrak{g}=\mathfrak{h}\oplus \mathfrak{m}$, where $\mathfrak{h}$ is the Lie subalgebra of $H$ and $\mathfrak{m}$ is a vector subspace invariant under the action of $H$ via adjoint representation $\Ad :G\to GL(\g)$.\\
We denote with $L_g:M\to M$ the group action of $G$ on $M$. This action defines a representation called \textit{linear isotropy representation} of $H$ in the linear group of $\T_oM$,
\begin{equation}
    \lambda: H \to GL(\T_oM);\ \ \  g\mapsto \lambda(g)=d_oL_g.
\end{equation}
In general, the action $L$ can be lifted to the orthonormal frame bundle $\OF{M}$ of $M$
\begin{align*}
    \tilde{L}:G\times P^{SO}_x(M) &\to P^{SO}_{L_g x}(M);\\
    (g,u) & \mapsto d_xL_g\circ u,
\end{align*}
where an element $u$ in the fiber over a point $x$ is a orientation-preserving linear isometry $u:\R^n\to \T_xM$. With this action, $G$ acts on $\OF{M}$ via automorphisms, hence $\OF{M}$ is a $G$-invariant principal $SO(n)$-bundle. Fixing a frame $u_o$ in the fiber over $o\in M$, the linear isotropy representation can be identified with the homomorphism $\lambda:H\to~SO(n)$ defined by
\[\lambda(g)=u_o^{-1}\circ d_oL_g\circ u_o.\]
When $M\simeq G/H$ is reductive, then $\lambda$ is faithful, and homogeneous metric-compatible connections exist. In our case, \textit{homogeneous metric-compatible connections} are connections, interpreted as $\so{n}$-valued 1-forms $\omega$ on the orthonormal frame bundle $\OF{M}$, that are $G$-invariant, namely
\[\tilde{L}_g^*\omega=\omega,\ \forall g\in G.\]
The existence of such a connection is ensured by a specialized version of Wang's theorem for reductive homogeneous spaces \cite{Kobayashi_Nomizu_1969}. Moreover, this Theorem also provides a classification of such connections.
\begin{theorem}[Wang's theorem]
    Let $P$ be a $G$-invariant principal $K$-bundle on a reductive homogeneous space $M\simeq G/H$ with decomposition $\mathfrak{g}=\mathfrak{h}\oplus \mathfrak{m}$. Then there is a one-to-one correspondence between the set of $G$-invariant connections on $P$ and the set of linear maps $\Lambda:\mathfrak{m}\to\mathfrak{k}$ such that 
    \[\Lambda(\Ad_h(v))=\Ad_{\lambda(h)}(\Lambda(v)),\ \text{ for all }  v\in\mathfrak{m},\, h\in H,\]
    where $\lambda$ denotes the isotropy homomorphism $H\to K$.
\end{theorem}
In the case of $\OF{M}$, we just need to set $K=SO(n)$. We would like to stress that the following vector spaces are isomorphic
\[\mathfrak{m}\xrightarrow{\sim} \T_oM \xrightarrow{\sim}  \R^n,\]
where the first isomorphism is given by the tangent map of the projection $\pi:G\to M$ on the identity element, while the second is by the choice of a frame $u_o$. Through those isomorphisms, the adjoint representation $\Ad:H\subset G \to GL(\g)$ is mapped to the linear isotropy representation, and in the isotropy homomorphism, respectively.

\medskip

We are interested in the cases in which $M$ admits a spin structure, namely a pair $(\SP{M}, \Bar{\rho})$ where $\SP{M}$ is a principal $\mathrm{Spin}(n)$-bundle over $M$ and $\Bar{\rho}:\SP{M}\to\OF{M}$ is a double-covering such that the following diagram commutes \cite{Bar_2024}.
\[
    \begin{tikzcd}
    \SP{M}\times \mathrm{Spin}(n) \arrow[r] \arrow[dd,"\Bar{\rho}\times\rho"] & \SP{M} \arrow[dd,"\Bar{\rho}"]\arrow[rd] & \\
       &   & M\\
    \OF{M}\times SO(n) \arrow[r] & \OF{M} \arrow[ru]& 
    \end{tikzcd}
\]
Here, $\rho:\mathrm{Spin}(n)\to SO(n)$ is the twofold covering homomorphism of $SO(n)$.
Over a homogeneous manifold $M\simeq G/H$, we call \textit{homogeneous spin structure} a spin structure $(\SP{M}, \Bar{\rho})$ equipped with an action of $G$ on $\SP{M}$ covering the action of $G$ on $\OF{M}$ and that is invariant under it. Namely, for every $g\in G$ there exists an automorphism $\Bar{L}_g:\SP{M}\to\SP{M}$ such that $\Bar{\rho}\circ \Bar{L}=\tilde{L}\circ \Bar{\rho}$.
In this case we can define a \textit{homogeneous spin connection} $\Bar{\omega}$ on $\SP{M}$ as the pullback of a homogeneous metric-compatible connection via a homogeneous spin structure $(\SP{M}, \Bar{\rho})$, i.e.\[\Bar{\omega}=\rho_*^{-1}\Bar{\rho}^*\omega.\]
However, a homogeneous spin structure does not always exist. To remedy this problem, we can notice that the set of homogeneous spin connections is in one-to-one correspondence, due to the injectivity of the pullback, with the set of homogeneous metric-compatible connections when a homogeneous spin structure exists. And, in the physical theory, the local fields do not depend on the specific spin structure, so they reflect directly the properties of homogeneous metric-compatible connections. Thus, we will study the latter set as a ``generalization'' of the first.

We recall that on the set of metric-compatible connections acts the group of vertical automorphisms, or gauge group, $\mathcal{G}$ of $\OF{M}$, whose elements are automorphisms of $\OF{M}$ that cover the identity on $M$.  A subgroup of $\mathcal{G}$ acts on the set of homogeneous metric-compatible connections, and it is defined by
\begin{equation}
\label{gau-inv}
    \mathcal{G}^G=\{f\in\mathcal{G}\ |\ f\circ \tilde{L}_g=\tilde{L}_g\circ f,\,\forall g\in G\}.
\end{equation}

On a homogeneous space $M\simeq G/H$, principal $\mathrm{Spin}(n)$-bundles have their own homogeneous connections for which homogeneous spin connections represent a subset. The collection of all possible homogeneous connections on all possible principal $\mathrm{Spin}(n)$-bundles modulo gauge transformations has been studied in the past years due to its relation with Yang-Mills theory. This collection is a moduli space defined as follows (cf. \cite{Biswas_Teleman_2014} for details)
\begin{equation}
    \mathcal{M}=\left\{\begin{matrix}(\mu,\Lambda)\in \mathrm{Hom}(H,\mathrm{Spin}(n))\times \mathrm{Hom}_{\R}(\mathfrak{g}/\mathfrak{h},\mathfrak{spin}(n))\ \\ \mbox{s.t.}\ \Lambda\circ \Ad_h=\Ad_{\mu(h)}\circ \Lambda\ \forall h\in H \end{matrix}\right\}/\mathrm{Spin}(n).
\end{equation}
where $\mathrm{Spin}(n)$ acts via conjugation on the pair $(\mu,\Lambda)$, namely the action is $(\mu,\Lambda)\mapsto(g\mu g^{-1},\Ad_g\circ\Lambda)$ for every $g\in\mathrm{Spin}(n)$.

\subsection{Specialty of dimension 3}
In the peculiar and physically relevant case of dimension 3, we want to analyze the collection of homogeneous spin connections and compare it with the moduli space of homogeneous $\mathrm{Spin}(3)$-connections.

\medskip

We recall a special feature of dimension $3$: simply connected homogeneous Riemannian manifolds are classified. It can be shown that the isotropy group $H$ is a Lie subgroup of $O(3)$, leading to three distinct cases according to the possible symmetries. Since $H$ is connected and compact, it must be isomorphic to $SO(3)$, $U(1)$, or be trivial. These cases correspond, respectively, to a homogeneous and isotropic Universe, to a homogeneous Universe with axial symmetry, and, finally, to the so-called Bianchi Universes. In what follows, we analyze these three cases separately.

\medskip

In dimension 3, we also have some useful isomorphisms. The spin bundle $\SP{M}$ is trivial, namely it is diffeomorphic to $M\times \mathrm{Spin}(3)$. Moreover, $\mathrm{Spin}(3)$ is isomorphic as Lie group to $SU(2)$, and the Lie algebras $\so{3}$ and $\su{2}$ are isomorphic. Furthermore,  $\su{2}$ equipped with the adjoint representation of $SU(2)$, $\so{3}$ equipped with the adjoint representation of $SO(3)$, and $\R^3$ equipped with the defining representation of $SO(3)$ are all isomorphic representations of $SU(2)$ considering the last two actions precomposed with $\rho:SU(2)\to SO(3)$. Namely, there exists an invertible equivariant map between each pair of those vector spaces.

\medskip

Using the previous properties, in the next Sections we are going to compute the moduli space of the two different types of connections in the three cosmological cases identified by the isotropy group $H$. As explained in Sec.\ref{sec:Intro}, the first moduli space is the moduli space of homogeneous metric-compatible connections. According to Wang's theorem, the set of homogeneous metric-compatible connections $\A^G$ is an affine space associated with the vector space
\begin{equation}
    \{\Lambda\in \mathrm{Hom}_{\R}(\R^3,\so{3})\ |\ \Lambda\circ \lambda(h)=\Ad_{\lambda(h)}\circ \Lambda\ \forall h\in H\}.
\end{equation}
The moduli space we are interested in is the space of the orbits $\A^G/\mathcal{G}^G$, where the dependence on the homogeneous space will be only encoded by $H$.
The second moduli space is the moduli space of homogeneous $SU(2)$ connections, defined by
\begin{equation}
    \label{moduli}
    \mathcal{M}=\left\{\begin{matrix}(\mu,\Lambda)\in \mathrm{Hom}(H,SU(2))\times \mathrm{Hom}_{\R}(\R^3,\su{2})\ \\ \mbox{s.t. }\Lambda\circ \lambda(h)=\Ad_{\mu(h)}\circ \Lambda\ \ \forall h\in H\end{matrix}\right\}/SU(2).
\end{equation}
Since $\g/\mathfrak{h}=\mathfrak{m}\cong\R^3$, even this moduli space depends only on the isotropy group $H$.

\section{Homogeneous Connections on Bianchi groups}
Let $M$ be a $3$-dimensional orientable connected Riemannian manifold equipped with a Lie group structure $G$ and with a left-invariant metric. It is a particular case of a Riemannian homogeneous space in which the stabilizer $H$ is trivial, $M\simeq G/\{e\}$. The simply connected spaces of this kind are known as Bianchi groups. On such spaces, the natural action we are going to consider is the left multiplication
\begin{align*}
    L:G\times G\to G;\\
    (g,g')\mapsto L_{g}g'=g g'
\end{align*}

\begin{remark}[\cite{Lawn_2022}]
    On a Lie group $G$, there exists a unique homogeneous spin structure $(\SP{G},\Bar{\rho})$ equipped with an action $\Bar{L}$ of $G$ that covers $\tilde{L}$.
\end{remark}
In this case, the space of homogeneous metric-compatible connections $\A^G$ is simply given by Wang's theorem as an affine space over the finite-dimensional vector space $\mathrm{Hom}_{\R}(\g,\so{3})$. 
The gauge group $\mathcal{G}^G$, which acts on $\A^G$, is the group of $G$-equivariant vertical automorphisms defined in \eqref{gau-inv}.
\begin{proposition}
\label{gaugehom}
    $\mathcal{G}^G$ is isomorphic to $SO(3)$.
\end{proposition}
\begin{proof}
    Let $p,p'\in\OF{G}$ be such that $\pi(p)=x$ and $\pi(p')=y$, then there exists a unique $g\in G$ such that $L_g x=y$. Hence, there exists a unique $a_{p,p'}\in SO(3)$ such that $\tilde{L}_gp=p'a_{p,p'}$. Let $f$ be in $\mathcal{G}^G$; the value of $f(p')$ is fixed by the value of $f(p)$. Recalling that there exists a unique $b_p\in SO(3)$ such that $f(p)=p\,b_p,\,\forall p\in\OF{G}$:
    \begin{equation}
        \begin{split}
        f(p')&=f(\tilde{L}_gp\,a^{-1}_{p,p'})=f(\tilde{L}_gp)a^{-1}_{p,p'}=\tilde{L}_gf(p)a^{-1}_{p,p'}=\\
        &=\tilde{L}_g(pb_p)a^{-1}_{p,p'}=\tilde{L}_gp\,b_pa^{-1}_{p,p'}=\\
        &=p'a_{p,p'}b_pa^{-1}_{p,p'}
        \end{split}
    \end{equation}
    From which $b_{p'}=a_{p,p'}b_pa^{-1}_{p,p'}$. Thus, the isomorphism is given, fixed a point $p\in\OF{G}$, by the map $f\mapsto b_p$.
\end{proof}
Thus, the space $\A^G/\mathcal{G}^G$, composed by equivalence classes of homogeneous metric-compatible connections, in reason of Wang's theorem and Prop.\ref{gaugehom}, is in one-to-one correspondence with $\mathrm{Hom}_{\R}(\R^3,\so{3})/SO(3)$ with the adjoint action of $SO(3)$ on its Lie algebra.

\medskip

Considering just the spin bundle $\SP{G}$, in \cite{Biswas_Teleman_2014}, we can find an expression of the moduli space $\mathcal{M}$ which can be interpreted as the disjoint union of the equivalence classes of homogeneous connections with respect to all the possible nonequivalent lifts of the $G$-action on the all possible spin bundles. In this particular case, the spin bundle can be only the trivial one, and all the lifts are equivalent. Indeed, the moduli space \eqref{moduli} is pretty simple $\mathcal{M}=\mathrm{Hom}_{\R}(\R^3,\su{2)}/SU(2)$ and the following proposition holds
\begin{proposition}
    $\A^G/\mathcal{G}^G$ and $\mathcal{M}$ are in one-to-one correspondence. Furthermore, there exists a bijective map from $\M$ to $M(3,\R)/SO(3)$, with the action of $SO(3)$ on the space of the $3\times 3$ real matrices $M(3,\R)$ given by the left composition.
\end{proposition}
\begin{proof}
    The Lie algebra $\so{3}$ and $\su{2}$ are isomorphic. Let $\rho:SU(2)\to SO(3)$ be the natural double covering homomorphism. The isomorphism between the Lie algebras is given by the pushforward $\rho_*:\su{2}\to\so{3}$ \cite{Lawson_Michelsohn_1989}. Considering $\so{3}$ and $\su{2}$ as carrying spaces of the respective adjoint action, the isomorphism $\rho_*$ is also an equivariant map, in the meaning that $\rho_*(\Ad_a v)=\Ad_{\rho(a)}\rho_*(v)$ for all $v\in\su{2},\,a\in SU(2)$.\\
    Elements of $\mathrm{Hom}_{\R}(\R^3,\su{2)}/SU(2)$ are equivalent classes of linear maps $[\Lambda]$ with $\Lambda\sim \Ad_a\circ\Lambda$ for any $a\in SU(2)$. We can define the map
    \begin{align*}
        \varrho:\mathrm{Hom}_{\R}(\R^3,\su{2)}/SU(2)&\to\mathrm{Hom}_{\R}(\R^3,\so{3})/SO(3);\\
        [\Lambda] &\mapsto [\rho_*\circ \Lambda].
    \end{align*}
    This map is well-defined thanks to the equivariant property 
    \begin{align*}
        \varrho[\Ad_a \circ \Lambda]=[\rho_*\circ\Ad_a \circ \Lambda]=[\Ad_{\rho(a)}\circ\rho_*\circ\Lambda]=[\rho_*\circ\Lambda]=\varrho[\Lambda].
    \end{align*}
    Clearly, it is invertible because $\rho_*$ is.

    \medskip

    Moreover, $\so{3}$ with the adjoint action of $SO(3)$ and $\R^3$ with the standard action of $SO(3)$ are isomorphic representation. The isomorphism is given by
    \begin{equation*}
        \begin{pmatrix}
            x\\
            y\\
            z
        \end{pmatrix}\mapsto\begin{pmatrix}
            0 & -z & y\\
            z & 0 & -x\\
            -y & x & 0
        \end{pmatrix}.
    \end{equation*}
    Hence, following a similar procedure, we conclude that there exists a one-to-one correspondence between $\mathrm{Hom}_{\R}(\R^3,\so{3})/SO(3)$ and $\mathrm{Hom}_{\R}(\R^3,\R^3)/SO(3)$.\\ Since $\mathrm{Hom}_{\R}(\R^3,\R^3)\cong M(3,\R)$, we conclude
    \begin{equation}
        \mathrm{Hom}_{\R}(\g,\su{2)}/SU(2)\xrightarrow{\sim}\mathrm{Hom}_{\R}(\g,\so{3})/SO(3)\xrightarrow{\sim}M(3,\R)/SO(3),
    \end{equation}
    with the action of $SO(3)$ on $M(3,\R)$ to be the left composition.
\end{proof}
Notice that the vector spaces associated to homogeneous connections are isomorphic as vector spaces $\mathrm{Hom}_{\R}(\g,\su{2)}\cong\mathrm{Hom}_{\R}(\g,\so{3})$. At the level of connections, this isomorphism can be implemented by the unique homogeneous spin structure, which maps homogeneous metric-compatible connections $\omega$ into homogeneous spin connections $\Bar{\omega}=\rho_*^{-1}\Bar{\rho}^*\omega$.

\bigskip

Our aim is characterize the moduli space $\mathcal{M}$ realized as $M(3,\R)/SO(3)$. Before to do that, let us notice that the set of positive (negative) semi-definite symmetric $3\times 3$ matrices $\mathcal{C}_{+}\,(\mathcal{C}_-)$ is a convex cone in the vector space of symmetric $3\times 3$ matrices, and its boundary $\partial\mathcal{C}_+\,(\partial\mathcal{C}_-)$ is composed by matrices with a vanishing determinant. Hence, we can enounce the following Lemma:
    \begin{lemma}
        The moduli space $M(3,\R)/SO(3)$ is in one-to-one correspondence with $\mathcal{C}_{+}\cup\mathcal{C}_-/\sim$. Where the equivalence relation is:
        \begin{equation}
            P,P'\in \partial\mathcal{C}_+\cup\partial\mathcal{C}_-,\ \ \ P\sim P'\ \mathrm{iff} \ P=-P'.
        \end{equation}
    \end{lemma}
\begin{proof}
    We consider $M(3,\R)$ equipped with the following action of $SO(3)$:
    \begin{align*}
        SO(3)\times M(3,\R)&\to M(3,\R);\\
        (O,M)&\mapsto OM.
    \end{align*}
    This action is clearly not free.\\
    Let $M\in M(3,\R)$, in reason of the polar decomposition, there exists a unique positive semi-definite symmetric $3\times 3$ matrix $P=\sqrt{M^tM}$ such that $M=UP$ for some $U\in O(3)$.\\
    If $\det(M)\neq 0$, then the orthogonal matrix is unique and $P$ is positive definite. Multiplying both for ${\rm sgn}(\det(M))$, we always obtain that $M$ is given by the composition of a special orthogonal matrix and a positive-definite (if $\det(M)>0$) or negative-definite (if $\det(M)<0$) symmetric matrix. Hence, the map $M\mapsto {\rm sgn}(\det(M))\sqrt{M^tM}$ is constant on the orbit of the $SO(3)$-action and differs on different orbits. Thus, the map passes to the quotient and defines a bijection between the orbit space of invertible matrices in $M(3,\R)$ and the set of positive-definite and negative-definite symmetric matrices.\\
    If $\det(M)=0$, we need to study separately the different ranks; in such cases, $U$ is no longer unique. If $\mathrm{rk}(M)=2$, let $U,U'\in O(3)$ be such that $M=UP=U'P$, then $(U-U')P=0$, and so $U=U'$ on $Ran(P)$. Furthermore, since $U,U'\in O(3)$ and $ker(P)\perp Ran(P)$, hence $U(ker(P))\subset ker(P)$ and $U'(ker(P))\subset ker(P)$, and then or $U= U'$, or $U=-U'$ on $ker(P)$. Thus, exactly one of $U$ and $U'$ is in $SO(3)$. This means that the stabilizer of $M$ is trivial. Therefore, for all $M$ such that $\mathrm{rk}(M)=2$, there exists a unique pair $U_+,U_-\in SO(3)$, and a unique positive and a unique negative semi-definite symmetric matrices $P_+$ and $P_-$ such that $M=U_+P_+=U_-P_-$. Where, $U_+,P_+$ is the pair constructed before, and $U_-=-\mathrm{Id}_{Ran(P_+)}U_+$ and $P_-=-\mathrm{Id}_{Ran(P_+)}P_+=-P_+$.\\
    If $\mathrm{rk}(M)=1$, we can proceed in a similar analysis as before. In this case, there are infinite $U,U'\in O(3)$; in fact, the stabilizer of such an $M$ is isomorphic to $SO(2)$. Moreover, there exist infinite pairs $U_+,U_-\in SO(3)$ as before, but the semi-definite symmetric matrices are unique: $P_+=\sqrt{M^tM}$ and $P_-=-P_+$.\\
    If $\mathrm{rk}(M)=0$, the orbit contains only the null matrix, and so $P_+=P_-=0$.\\
    This means that, when $\det(M)=0$, in the same orbit there are $P_+$ and $P_-$. So, it is no longer true that different semi-definite symmetric matrices identify different orbits. Thus, we need to impose the following equivalence relation on semi-definite symmetric matrices with $\det(P)=0$: $P\sim P'$ if and only if $P=-P'$.\\
\end{proof}

Let us consider $\M$ as in the above identification, and the space of traceless symmetric $3\times 3$ matrices ${\rm Sym}_0(\R^3)$. Both spaces admit a filtration in terms of algebraic manifolds with edges. Consider $\M$ with collection of subspaces \[M_n=\{P\in M(3,\R)/SO(3)\ |\ {\rm rk}(P)\leq n\} \quad \mbox{for } n=0,\dots,3.\]
While, for ${\rm Sym}_0(\R^3)\times \R$, we can consider the subspaces $S_n$ given by \begin{gather*}
    S_3={\rm Sym}_0(\R^3)\times \R,\\ 
    S_2=\left\{(A,\lambda)\in {\rm Sym}_0(\R^3)\times \R\ |\ \lambda=0\right\},\\
    S_1=\left\{(A,\lambda)\in{\rm Sym}_0(\R^3)\times \R\ \Big|\ \begin{matrix}
        27(\det(A))^2=-4({\rm min}(A))^3\\
        \lambda=0,\,\det(A)\geq0
    \end{matrix}\right\},\\ 
    S_0=\left\{(A,\lambda)\in{\rm Sym}_0(\R^3)\times \R\ |\ \lambda=0,\,A=0\right\}.
\end{gather*}
Where ${\rm min} (A)$ denotes the second coefficient of the characteristic polynomial, namely: 

\[{\rm min}(A)=\frac{1}{2}\left(tr(A^2)-(trA)^2\right),\ \text{ for $A\in M(3,\R)$.}\]


This filtration induces a stratification of the spaces, which on $\M$ coincides with the stratification naturally induced by the quotient. Moreover, the equation that describes $S_1$ is extremely singular; in fact, if there were a point on which the gradient of the equation is non-trivial, then it would be possible to describe $S_1$ locally as a four-dimensional manifold. This is impossible because the choice of an element in $S_1$ is equivalent to the choice of an eigenvalue $2\mu >0$ and the relative eigenvalue, and those conditions can be done with three degrees of freedom. So, as a topological manifold, $S_1$ has dimension 3 while $S_2$ has dimension 5. As a result, we can characterize the moduli space $\M$ in terms of traceless symmetric matrices:

\begin{proposition}
    There exists a weakly stratified isomorphism between $\mathcal{M}$ and ${\rm Sym}_0(\R^3)\times \R$.
\end{proposition}
\begin{proof}
     To provide the topological manifold structure, we specify a bijective map $\phi:\M\to{\rm Sym}_0(\R^3)\times \R$, which, using the isomorphism ${\rm Sym}_0(\R^3)\times \R\cong \R^6$, gives us a global chart for $\M$. Specifically, we are going to construct a global chart for $\mathcal{C}_{+}$ and $\mathcal{C}_-$, showing that the gluing of them gives us a global chart for the whole $\mathcal{M}$. We define the following maps
    \begin{align*}
        \phi_+:\mathcal{C}_{+}&\to {\rm Sym}_0(\R^3)\times \R_+;\\
        P&\mapsto \left(P-\tfrac{1}{3}tr(P)\mathrm{Id}_3, \lambda_P\right),
    \end{align*}
    where $\lambda_P$ is the eigenvalue of $P$ with the smallest modulus $|\lambda_P|$. This map is bijective with an inverse
    \begin{align*}
        \phi_+^{-1}:{\rm Sym}_0(\R^3)\times \R_+&\to\mathcal{C}_{+};\\
        (A,\lambda)&\mapsto A+(\lambda-\mu_A)\mathrm{Id}_3,
    \end{align*}
    where $A$ is a null-trace symmetric $3\times 3$ matrix and $\mu_A$ is its smallest eigenvalue (surely non-positive). Hence, this map defines a global chart for $\mathcal{C}_{+}$.\\
    With a change of sign, we obtain the global chart for $\mathcal{C}_{-}$
    \[\begin{matrix}
        \phi_-:&\mathcal{C}_{-}&\to &{\rm Sym}_0(\R^3)\times \R_-;\\
        &P&\mapsto &\left(-P+\tfrac{1}{3}tr(P)\mathrm{Id}_3, \lambda_P\right).\\
        & & & \\
        \phi_-^{-1}:&{\rm Sym}_0(\R^3)\times \R_-&\to&\mathcal{C}_{-};\\
        &(A,\lambda)&\mapsto& -A+(\lambda+\mu_A)\mathrm{Id}_3.
    \end{matrix}\]
    Remains to show that the gluing of these two maps is well-defined. Considering $P\in\partial\mathcal{C}_{+}$ and so $-P\in\partial\mathcal{C}_{-}$, hence $\lambda_P=0=\lambda_{-P}$, thus
    \[\phi_-(-P)=\left(-(-P)+\tfrac{1}{3}tr(-P)\mathrm{Id}_3,0\right)=\left(P-\tfrac{1}{3}tr(P)\mathrm{Id}_3,0\right)=\phi_+(P).\]
    We can check that the projection to the quotient $\pi:M(3,\R)\to\mathcal{M}$ is continuous with respect to the global chart $(\mathcal{M},\phi)$. Let us consider the map $M\mapsto\mathrm{sgn}(\det(M))\sqrt{M^tM}$, where $\mathrm{sgn}(\det(M))$ can be $\pm1$ indifferently when $\det(M)$ vanishes, and noticing $\lambda_{-P}=-\lambda_P$ in general, we obtain
\begin{equation}
    \phi\circ \pi(M)=\left(\sqrt{M^tM}-\tfrac{1}{3}tr(\sqrt{M^tM})\mathrm{Id}_3,\,\mathrm{sgn}(\det(M))\lambda_{\sqrt{M^tM}}\right),
\end{equation}
    which is a continuous map.

    \medskip

    We now prove that the strata are mapped appropriately. Evidently, $\phi(M_3)=S_3$ and $\phi(M_0)=S_0$. Moreover, it is immediate to see that $\phi(M_2)\subset S_2$. Conversely, ${\rm rk}(\phi^{-1}(A))\leq 2$, indeed, the eigenvector of $A$ with the smallest eigenvalues $\mu_A$ is in the kernel of $\phi^{-1}(A)=A-\mu_A\mathrm{Id}_3$, then $\phi(M_2)= S_2$.
    
    Let us now consider $P\in M_1$, given a basis of eigenvectors $v_0,v_1,v_2$, with $v_1,v_2\in{ ker}(P)$ and $v_o$ with eigenvalues $\lambda\geq 0$ (without loss in generality), they are eigenvectors of $\phi(P)=P-\frac{1}{3}tr(P)\mathrm{Id}_3$ with eigenvalues $-\frac{\lambda}{3}$ and $\frac{2\lambda}{3}$, respectively. Hence, $\det(\phi(P))=\frac{2}{27}\lambda^3\geq 0$ and ${\rm min}(\phi(P))=-\frac{\lambda^2}{3}$, from which $27(\det(A))^2=-4({\rm min}(A))^3$. Then, $\phi(M_1)\subset S_1$\\
    Let $A\in S_1$, its eigenvalues $\mu_1,\mu_2,\mu_3$ satisfy
    \begin{equation*}
        \mu_1+\mu_2+\mu_3=0,\,\text{ and }\ \tfrac{1}{4}(\mu_1\mu_2\mu_3)^2=-\tfrac{1}{27}(\mu_1\mu_2+\mu_2\mu_3+\mu_3\mu_1)^3
    \end{equation*}
    Plugging the first into the second, we obtain an equation
    \[4\mu_1^6+12\mu_1^5\mu_2-3\mu_1^4\mu_2^2-26\mu_1^3\mu_2^3-3\mu_1^2\mu_2^4+12\mu_1\mu_2^5+4\mu_2^6=0.\]
    A possible solution is $\mu_2=0$, in this case $\det(A)=0$, and so ${\rm min}(A)=0$ and $tr(A)=0$, imposing that $A=0$. Considering $\mu_2\neq 0$, we can introduce the variable $t=\mu_1/\mu_2$, and the previous equation reads
    \[4t^6+12t^5-3t^4-26t^3-3t^2+12t+4=0.\]
    This polynomial equation has three roots, each with multiplicity $2$: $t=1$, $t=-2$, $t=-\frac{1}{2}$, providing three similar cases: $\mu_1=\mu_2=-\frac{1}{2}\mu_3$, $\mu_2=\mu_3=-\frac{1}{2}\mu_1$, and $\mu_3=\mu_1=-\frac{1}{2}\mu_2$. The condition $\det(A)>0$ forces the two equal eigenvalues to be negative. Thus, the spectrum is 
    \[\sigma(A)=\{-\mu,-\mu,2\mu\}\ \text{ with }\ \mu\geq 0.\]
    Now, $\phi^{-1}(A)$ is given by $A+\mu\mathrm{Id}_3$ which has rank almost $1$. Thus, $\phi(M_1)=S_1$.
\end{proof}

As part of the proof, we get the following corollary
\begin{corollary}
    The moduli space $\M$ is a topological manifold homeomorphic to $\R^6$.
\end{corollary}

\section{Axial symmetry: a peculiar case}
This case represents a peculiar symmetry in physics, characterized by the stabilizer being the abelian group \( U(1) \). The peculiarity of this case is that there are two possible inequivalent choices for the isotropy homomorphism; however, they provide the same classification of homogeneous connections.\\
Since we do not fix the group $G$ nor the action on $M$, we can figure out the isotropic homomorphism by looking for injective homomorphisms in the set of inequivalent homomorphisms $\mathrm{Hom}(U(1), SO(3))/SO(3)$, with $SO(3)$ acting via conjugation. This classification is provided by the maps into the maximal tori of $SO(3)$. Since the maximal tori are conjugate to each other, we can pick a representative as the map into the natural $SO(2)$ subgroup. The only two elements of the set represented by injective maps are 
\begin{align}
\nonumber
     \lambda_{\pm}:&U(1)\to SO(3);\\
    &e^{\pm i\theta}\mapsto \begin{pmatrix}
        1 & 0 & 0\\
        0 & \cos{\theta} & -\sin{ \theta}\\
        0 & \sin{ \theta} & \cos{\theta}
    \end{pmatrix}
\end{align}
Using the equivariant map between $\so{3}$ and $\R^3$ we obtain two vector spaces associated to homogeneous metric-compatible connections, one for each isotropy homomorphism 
\[\vec{\A}_{\pm}^{ax}=\{\Lambda\in \mathrm{Hom}(\R^3,\R^3)\ |\ \Lambda\circ \lambda_{\pm}(h)=\lambda_{\pm}(h)\circ\Lambda,\text{ for all }  h\in H\}.\]
Both these vector spaces read as the space of linear maps on $\R^3$ that commute with the rotation about the $x$ axis. Such matrices are the multiples of rotational matrices on the plane $(y,z)$, thus
\begin{equation}
\label{connAx}
    \A^{ax}\cong \left\{\begin{pmatrix}
        a & 0 & 0\\
        0 & b & -c\\
        0 & c & b
    \end{pmatrix}\ \text{ s.t. }\ a,b,c\in\R \right\}\cong \R^3.
\end{equation}
As anticipated, the group $\mathcal{G}^G$ is independent of the group $G$ and depends only on the stabilizer $H=U(1)$. Then, we call the group in this axial symmetric case $\mathcal{G}^{ax}$.
\begin{proposition}
\label{gaugeax}
    $\mathcal{G}^{ax}$ is independent of the choice of isotropy homomorphism, and it is isomorphic to $SO(2)$.
\end{proposition}
\begin{proof}
    The proof is similar to the proof of Prop.\ref{gaugehom}. In this case, we have a preferred point $p_0^{\pm}\in\OF{M}$ in the fiber over $o\in M$ for each choice of the isotropic homomorphism $\lambda_{\pm}$, such that $\tilde{L}_hp_0^{\pm}=p_0^{\pm}\lambda_{\pm}(h)$ for all $h\in H$. Considering the homomorphism $f\mapsto b_{p_0^{\pm}}$, where $f(p_0^{\pm})=p_0^{\pm}\,b_{p_0^{\pm}}$, we have a constraint on the values of $b_{p_0^{\pm}}$:
    \begin{align*}
        p_0^{\pm}\, b_{p_0^{\pm}}\lambda_{\pm}(h)=f(\tilde L_h p_0^{\pm})=\tilde L_h f(p_0^{\pm})=p_0^{\pm} \lambda_{\pm}(h)b_{p_0^{\pm}},
    \end{align*}
    namely $b_{p_0^{\pm}}$ must commute with $\lambda_{\pm}(h)\in SO(3)$ for any $h\in H$. Thus, $b_{p_0^{\pm}}\in \lambda_{\pm}(H)\cong SO(2)$ because $\lambda_{\pm}(H)$ is a maximal torus in $SO(3)$.
\end{proof}
As a consequence, $\mathcal{G}^{ax}$ acts via the adjoint representation of $\lambda(H)\cong SO(2)\subset SO(3)$ on the $\so{3}$ factor, namely via left matrix multiplication on $\A^{ax}$ as expressed in \eqref{connAx}. From which, we can compute the moduli space:
\begin{proposition}
    $\A^{ax}/\mathcal{G}^{ax}$ is a smooth manifold with boundary diffeomorphic to $\R\times\R_{\geq0}.$
\end{proposition}
\begin{proof}
    Using polar coordinates $b=r\cos{\theta},c=r\sin{\theta}$, every element of $\A^{ax}$ can be written as 
    \[\begin{pmatrix}
        a & 0 & 0\\
        0 & r\cos{\theta} & -r\sin{\theta}\\
        0 & r\sin{\theta} & r\cos{\theta}
    \end{pmatrix},\ a,\theta\in \R,\,r\in\R_{\geq 0}.\]
    Varying $\theta$, we move along an orbit, and we can clearly see that every orbit has a different representative for $\theta=0$, obtaining:
    \[\A^{ax}/\mathcal{G}^{ax}\cong \left\{\begin{pmatrix}
        a & 0 & 0\\
        0 & r & 0\\
        0 & 0 & r
    \end{pmatrix}\ |\ a\in\R,\,r\in\R_{\geq 0} \right\}\cong \R\times \R_{\geq 0}\]
\end{proof}
Nevertheless, on such a homogeneous space, a homogeneous spin structure does not always exist (see \cite{Lawn_2022} for the case $U(2)/U(1)$). Hence, a proper notion of homogeneous spin connection can not be found. However, we can look at the space of homogeneous connections on the spin bundle $\SP{M}$.

\medskip

The formula \ref{moduli} gives us a rapid way to compute the moduli space of homogeneous connections on a spin bundle. First of all, we need to compute the quotient space $\mathrm{Hom}(U(1),SU(2))/SU(2)$, which counts the inequivalent lifts on $\SP{M}$ of the $G$-action on $M$. The classification is analogous to the $SO(3)$ homomorphism. Looking at a fixed maximal torus in $SU(2)$ to find the representatives
\begin{equation}
\begin{matrix}
        \mu_n:& U(1)&\to &SU(2);\\
        &e^{i\theta}& \mapsto&\begin{pmatrix}
        e^{-in\theta} & 0\\
        0 & e^{in\theta}
    \end{pmatrix},
    \end{matrix}
    \end{equation}
we describe countable distinguishable classes, each one represented by $\mu_n$ with $n\in\mathbb{Z}$.
We can now compute the moduli space of homogeneous $SU(2)$ connections in this axially symmetric case.
\begin{proposition}
    The moduli space $\mathcal{M}^{ax}$ has countable connected components, each diffeomorphic to $\R$.
\end{proposition}
\begin{proof}
    We recall that we have a priori two moduli spaces, each associated with a possible isotropy homomorphism
    \[\mathcal{M}^{ax}_{\pm}=\left\{\begin{matrix}(\mu,\Lambda)\in \mathrm{Hom}(U(1),SU(2))\times \mathrm{Hom}_{\R}(\R^3,\su{2})\ \\ \mbox{s.t. }\ \Lambda\circ \lambda_{\pm}(h)=\Ad_{\mu(h)}\circ \Lambda \  \forall h\in H\end{matrix}\right\}/SU(2).\]
    Using the isomorphism between $\su{2}$ and $\R^3$, we can construct the action $\mu_n(U(1))\subset SU(2)$ on $\R^3$:
    \begin{align*}
    &    \begin{pmatrix}
        e^{-in\theta} & 0\\
        0 & e^{in\theta}
    \end{pmatrix}\begin{pmatrix}
        ix & -y+iz \\
        y+iz & -ix
    \end{pmatrix}\begin{pmatrix}
        e^{in\theta} & 0\\
        0 & e^{-in\theta}
    \end{pmatrix}=\\
    &=\begin{pmatrix}
        ix & (-y+iz)e^{-2ni\theta} \\
        (y+iz)e^{2ni\theta} & -ix
    \end{pmatrix}.
    \end{align*}
    Hence, $\mu_n(U(1))$ acts on the vector $(x,y,z)$ as a rotational matrix \[R(2n\theta)=\begin{pmatrix}
        1 & 0 & 0\\
        0 & \cos{2n\theta} & -\sin{2n\theta}\\
        0 & \sin{2n\theta} & \cos{2n\theta}
    \end{pmatrix}.\]
    This means that, for a fixed $\mu_n$ and $\lambda_{\pm}$, we are looking for endomorphisms $\Lambda$ of $\R^3$ such that $\Lambda R(\pm\theta)=R(2n\theta)\Lambda$. Those $\Lambda$ must be like
    \[\Lambda=\begin{pmatrix}
        c & 0 & 0\\
        0 & 0 & 0\\
        0 & 0 & 0
    \end{pmatrix},\ c\in\R.\]
    Clearly the description is the same for $\lambda_{+}$ or $\lambda_{-}$. This is enough to provide the classification. Considering an equivalent class $[\mu,\Lambda']$, there exists an $a\in SU(2)$ and $n\in \mathbb{Z}$ such that $[\mu,\Lambda']=[\mu_n,\Ad_a\circ\Lambda']$, where $\Ad_a\circ\Lambda'$ is in the type of $\Lambda$ found before. If there exists $a\in SU(2)$ such that $[\mu_n,\Lambda]=[\mu_n,\Ad_a\circ\Lambda]$, then $a$ can be associated to a special orthogonal matrix $O_a$ such that, when applied to $\Lambda$, gives a matrix $\Lambda'=O_a\Lambda$ with all entries zero but the first. However, considering the first column as a vector, it must conserve its modulus, hence the first entry must be the same as $\Lambda$, then $O_a\Lambda=\Lambda$.\\
    In conclusion, every class of homomorphism defines a connected component diffeomorphic to $\R$.
\end{proof}

Here, we observe a clear distinction between the two moduli spaces. This significant difference from the purely homogeneous case arises from the non-existence of a homogeneous spin structure, which prevents the embedding of homogeneous metric-compatible connections into the $SU(2)$ framework. This obstruction is due to the failure of the $G$-action to lift from the orthonormal frame bundle to the spin bundle. In the next section, it will become evident which of the two moduli spaces corresponds to the classification of homogeneous Ashtekar–Barbero–Immirzi–Sen connections as required by the physical theory.

\section{Isotropic spaces}
This class of symmetry is one of the most studied in physics; it is characterized by a stabilizer $H=SO(3)$. The most relevant cases consist of the three simply connected homogeneous spaces $(\R^3,\mathbb{E}_0(3))$, $(\mathbb{S}^3,SO(4))$, and $(\mathbb{H}^3,SO^+(1,3))$. However, as in the previous cases, we can provide a unique treatment dependent only on the stabilizer subgroup.

\medskip

The classification of homogeneous metric-compatible connections starts by noticing that the isotropic homomorphism $\lambda:H\to SO(3)$ is just the identity. Indeed, since $SO(3)$ has no normal subgroup, a homomorphism $\lambda:SO(3)\to SO(3)$ is trivial or bijective. But the isotropic homomorphism must be injective for a reductive homogeneous space. Moreover, all the automorphisms of $SO(3)$ are inner automorphisms, so they are conjugate to each other. Hence, a suitable choice of $p_0$ in the orthonormal frame bundle give us $\lambda=\mathrm{id}_{SO(3)}$. 
Using the version of Wang's theorem for reductive homogeneous spaces, we can compute the set of homogeneous metric-compatible connections $\A^{iso}$, and for any isotropic homogeneous space, we get
\begin{equation}
    \A^{iso}\cong\left\{\Lambda\in\mathrm{Hom}_{\R}(\R^3,\so{3})\ |\ \Lambda\circ R=\Ad_{R}\circ \Lambda,\,\forall R\in SO(3)\right\}.
\end{equation}
The invariant gauge group $\mathcal{G}^{iso}$ is simplified a lot on this class of spaces
\begin{lemma}
    $\mathcal{G}^{iso}$ contains only the identity element
\end{lemma}
\begin{proof}
    The proof follows from the computation done for Prop.\ref{gaugeax}. Since $b_{p_0}$ must commute with $\lambda(SO(3))$ and $\lambda$ is bijective, $b_{p_0}$ must leave in the center of $SO(3)$, which is trivial. Hence, $b_{p_0}=e$, and so $b_p=e$ for all $p\in P^{SO}$.
\end{proof}
\begin{proposition}
    $\A^{iso}/\mathcal{G}^{iso}$ is diffeomorphic to $\R$
\end{proposition}
\begin{proof}
    We already discussed that $\mathcal{G}^{iso}$ contains only the identity element. Hence, we just need to discuss $\A^{iso}$. Using the equivariant isomorphism between $\so{3}$ and $\R^3$, we reduce to discussing the linear endomorphisms of $\R^3$ that commute with the rotational matrices, namely, we are looking for $\Lambda:\R^3\to\R^3$ such that $\Lambda\circ R=R\circ \Lambda$ for all $R\in SO(3)$. It is easy to show that they are only multiples of the identity. Thus, the vector space associated to $\A^{iso}$ is canonically $\R$.
\end{proof}
As in the axial-symmetric case, on such spaces does not always exist an invariant homogeneous structure (see \cite{Lawn_2022} for the case $SO(4)/SO(3)$). Hence, we can not define a proper homogeneous spin structure. However, in this case, there is only one homogeneous connection on the spin bundle.
\begin{proposition}
    The moduli space of homogeneous $SU(2)$ connections for isotropic spaces is trivial.
\end{proposition}
\begin{proof}
    Let us start from the computation of $\mathrm{Hom}(SO(3),SU(2))$. Such a homomorphism cannot be injective, but the kernel can be only trivial or the whole $SO(3)$. Thus, there exists only one homomorphism, which is the trivial one.\\
    The moduli space formula \eqref{moduli} says us, independently by the specific space \[\mathcal{M}^{iso}=\left\{\Lambda\in \mathrm{Hom}_{\R}(\R^3,\su{2})\ |\ \Lambda\circ \lambda(h)=\Lambda\ \forall h\in SO(3)\right\}/SU(2).\]
    Using the equivariant isomorphism between $\su{2}$ and $\R^3$, we reduce to discuss the linear endomorphisms of $\R^3$ such that $\Lambda\circ R=\Lambda$ for all $R\in SO(3)$. Clearly, the only possible solution is $\Lambda=0$. Thus, the moduli space contains only the canonical connection.
\end{proof}

The isotropic case provides an opportunity to compare the two spaces with the standard treatment of cosmology in the Ashtekar-Barbero-Immirzi formulation. This formulation has been extensively studied in the context of Loop Quantum Cosmology, as it represents the starting point of the classical framework. In \cite{Ashtekar_Bojowald_Lewandowski_2003}, it is stated that the homogeneous and isotropic Ashtekar-Barbero-Immirzi-Sen connections are classified by a gauge-invariant parameter, which is simply a real number. Then, the moduli space of homogeneous and isotropic Ashtekar-Barbero-Immirzi-Sen connections is just the real line \cite{Brunnemann_Fleischhack_2012}. This classification aligns with the moduli space of homogeneous metric-compatible connections, which become spin connections upon fixing a spin structure, while it automatically excludes the $SU(2)$ connection interpretation, as the latter is represented by a single point.

\section{Conclusion}
We have classified both the homogeneous metric-compatible connections and the homogeneous $SU(2)$ connections on three-dimensional Riemannian homogeneous spaces. In particular, we computed the moduli space resulting from the action of the gauge group on the space of homogeneous connections, and we found that its structure depends only on the stabilizer subgroup of the homogeneous space. Furthermore, the moduli space of homogeneous metric-compatible connections exhibits favorable topological properties: it is a finite-dimensional topological manifold (possibly with boundary) and is contractible.

\medskip

This analysis leads to a significant conclusion: the correct interpretation of the homogeneous Ashtekar–Barbero–Immirzi–Sen connection is that of a homogeneous spin connection. Only in this setting does the resulting structure align with the established physical literature, especially in the isotropic case.

These features have direct consequences for the mathematical foundations of the quantum theory. Most notably, they guarantee the existence of a regular Borel measure and allow the use of the Riesz–Markov–Kakutani representation theorem, thereby facilitating a rigorous construction of gauge-invariant quantum states. Additionally, the contractibility of the moduli space implies the triviality of the fiber bundle $\mathcal{\A}^G \to \A^G/\mathcal{G}^G$, enabling a global treatment of the gauge-fixing procedure without encountering topological obstructions.

\section*{Acknowledgments}
The authors also gratefully acknowledge financial support from Sapienza Università di Roma within Progetto di Avvio alla Ricerca di Ateneo 2023 (grant no. AR1231884E22A3E0) and Progetto di ricerca 2023 (grant no. RM123188 F1F08F61). This work has been carried out under the auspices of the GNFM INdAM (Gruppo Nazionale per la Fisica Matematica-Istituto Nazionale di Alta Matematica), and within the framework of the activities for PRIN 2022AKRC5P “Interacting Quantum Systems: Topological Phenomena and Effective Theories”. Project PE0000023-NQSTI.

\bibliographystyle{unsrt} 
\bibliography{biblio}

\begin{thebibliography}{10}

\bibitem{Thurston_1982}
William~P. Thurston.
\newblock Three dimensional manifolds, kleinian groups and hyperbolic geometry.
\newblock {\em Bulletin of the American Mathematical Society}, 6(3):357–381, 1982.

\bibitem{Landau_Lifshits_1975}
Lev~Davidovich Landau and Evgeniĭ~Mikhaĭlovich Lifshits.
\newblock {\em The Classical Theory of Fields: Volume 2}.
\newblock Butterworth-Heinemann, 1975.

\bibitem{Milnor_1976}
John Milnor.
\newblock Curvatures of left invariant metrics on lie groups.
\newblock {\em Advances in Mathematics}, 21(3):293–329, September 1976.

\bibitem{Patrangenaru_1996}
Victor Patrangenaru.
\newblock Classifying 3- and 4-dimensional homogeneous riemannian manifolds by cartan triples.
\newblock {\em Pacific Journal of Mathematics}, 173(2):511–532, April 1996.

\bibitem{Ashtekar_1986}
Abhay Ashtekar.
\newblock New variables for classical and quantum gravity.
\newblock {\em Physical Review Letters}, 57(1818):2244–2247, November 1986.

\bibitem{Ashtekar_1987}
Abhay Ashtekar.
\newblock New hamiltonian formulation of general relativity.
\newblock {\em Physical Review D}, 36(66):1587–1602, September 1987.

\bibitem{Barbero_1995}
J.~Fernando Barbero~G.
\newblock Real ashtekar variables for lorentzian signature space-times.
\newblock {\em Physical Review D}, 51(1010):5507–5510, May 1995.

\bibitem{Immirzi_1997}
Giorgio Immirzi.
\newblock Real and complex connections for canonical gravity.
\newblock {\em Classical and Quantum Gravity}, 14(1010):L177, October 1997.

\bibitem{Bruno_2024a}
Matteo Bruno.
\newblock Fiber bundle structure in ashtekar-barbero-immirzi formulation of general relativity.
\newblock {\em Journal of Geometry and Physics}, 214:105537, August 2025.

\bibitem{Bruno_2024b}
Matteo Bruno.
\newblock Cosmology in loop quantum gravity: symmetry reduction preserving gauge degrees of freedom.
\newblock {\em Classical and Quantum Gravity}, 42(18):185009, September 2025.

\bibitem{Wang_1958}
Hsien-Chung Wang.
\newblock On invariant connections over a principal fibre bundle.
\newblock {\em Nagoya Mathematical Journal}, 13:1–19, June 1958.

\bibitem{Bojowald_Kastrup_2000}
Martin Bojowald and Hans~A. Kastrup.
\newblock Quantum symmetry reduction for diffeomorphism invariant theories of connections.
\newblock {\em Classical and Quantum Gravity}, 17(1515):3009–3043, August 2000.
\newblock arXiv:hep-th/9907042.

\bibitem{Kobayashi_Nomizu_1969}
Shoshichi Kobayashi and Katsumi Nomizu.
\newblock {\em Foundations of differential geometry. Vol. II}, volume~15 of {\em Intersci. Tracts Pure Appl. Math.}
\newblock Interscience Publishers, New York, NY, 1969.

\bibitem{Bar_2024}
Christian B\"ar.
\newblock Spin geometry, March 2018.
\newblock \url{https://www.math.uni-potsdam.de/fileadmin/user_upload/Prof-Geometrie/Dokumente/Lehre/Veranstaltungen/SS11/spingeo.pdf}.

\bibitem{Biswas_Teleman_2014}
Indranil Biswas and Andrei Teleman.
\newblock Invariant connections and invariant holomorphic bundles on homogeneous manifolds.
\newblock {\em Central European Journal of Mathematics}, 12(1):1–13, January 2014.

\bibitem{Lawn_2022}
Jordi Daura~Serrano, Michael Kohn, and Marie-Amélie Lawn.
\newblock G-invariant spin structures on spheres.
\newblock {\em Annals of Global Analysis and Geometry}, 62(22):437–455, September 2022.

\bibitem{Lawson_Michelsohn_1989}
H.~Blaine Lawson and Marie-Louise Michelsohn.
\newblock {\em Spin Geometry (PMS-38)}.
\newblock Princeton University Press, 1989.

\bibitem{Ashtekar_Bojowald_Lewandowski_2003}
Abhay Ashtekar, Martin Bojowald, and Jerzy Lewandowski.
\newblock Mathematical structure of loop quantum cosmology.
\newblock {\em Advances in Theoretical and Mathematical Physics}, 7(22):233–268, March 2003.

\bibitem{Brunnemann_Fleischhack_2012}
Johannes Brunnemann and Christian Fleischhack.
\newblock Non-almost periodicity of parallel transports for homogeneous connections.
\newblock {\em Mathematical Physics, Analysis and Geometry}, 15(4):299–315, December 2012.

\end{thebibliography}

\end{document}